\documentclass{vldb}
\makeatletter
\let\old@citex\@citex
\makeatother
\usepackage[english]{babel}
\makeatletter
\let\@citex\old@citex
\makeatother
\usepackage{mathptmx,helvet}

\usepackage{times,balance,graphicx,amssymb,amsmath,subfigure,url}
\usepackage{comment}

\newtheorem{lemma}{Lemma}

\begin{document}

\title{Efficient Processing of Very Large Graphs in a Small Cluster}

\author{
{Da Yan$^1$,\ \ Yuzhen Huang$^2$,\ \ James Cheng$^3$,\ \ Huanhuan Wu$^4$}
\vspace{1.6mm}\\
\fontsize{10}{10}\selectfont\itshape\rmfamily Department of Computer Science and Engineering, The Chinese University of Hong Kong\\
\fontsize{9}{9}\selectfont\ttfamily\upshape \{$^1$yanda, $^2$yzhuang, $^3$jcheng, $^4$hhwu\}@cse.cuhk.edu.hk
}

\maketitle

\begin{abstract}
Inspired by the success of Google's Pregel, many systems have been developed recently for iterative computation over big graphs. These systems provide a user-friendly vertex-centric programming interface, where a programmer only needs to specify the behavior of one generic vertex when developing a parallel graph algorithm. However, most existing systems require the input graph to reside in memories of the machines in a cluster, and the few out-of-core systems suffer from problems such as poor efficiency for sparse computation workload, high demand on network bandwidth, and expensive cost incurred by external-memory join and group-by.

In this paper, we introduce the GraphD system for a user to process very large graphs with ordinary computing resources. GraphD fully overlaps computation with communication, by streaming edges and messages on local disks, while transmitting messages in parallel. For a broad class of Pregel algorithms where message combiner is applicable, GraphD eliminates the need of any expensive external-memory join or group-by. These key techniques allow GraphD to achieve comparable performance to in-memory Pregel-like systems without keeping edges and messages in memories. We prove that to process a graph $G=(V, E)$ with $n$ machines using GraphD, each machine only requires $O(|V|/n)$ memory space, allowing GraphD to scale to very large graphs with a small cluster. Extensive experiments show that GraphD beats existing out-of-core systems by orders of magnitude, and achieves comparable performance to in-memory systems running with enough memories.
\end{abstract}

\section{Introduction}\label{sec:intro}
Google's Pregel~\cite{pregel} and Pregel-like systems (e.g., Giraph~\cite{giraph}, GraphLab~\cite{graphlab,powergraph} and Pregelix~\cite{pregelix}) have become popular for iterative computation over big graphs recently, with numerous applications including social network analysis~\cite{socialnet}, webpage ranking~\cite{pregel}, and graph matching~\cite{match1,match2}. These systems provide a user-friendly programming model, where a user thinks like a vertex when developing a parallel graph algorithm.

However, most existing systems require an entire input graph to reside in memories, as well as the huge amounts of messages generated during the computation. While this assumption is proper for big companies and researchers with powerful computing resources, it neglects the need of an average user who wants to process very large graphs, such as small businesses and researchers that cannot afford a large cluster. For example, \cite{pregelix} reported that in the Giraph user mailing list there are 26 cases (among 350 in total) of out-of-memory related issues from March 2013 to March 2014.

Another problem with in-memory systems is that, they often use much more memory than the actual size of the input graph in order to keep the supporting data structures, vertex states, edges and messages. For example, \cite{mocgraph} reported that to process a graph dataset that takes only 28GB disk space, Giraph and GraphLab need 370GB and 800GB memory space, respectively; and when memory resources become exhausted, the performance of Giraph degrades seriously while GraphLab simply crashes. Thus, even using a cluster with terabytes of memory space, we may only be able to process a graph of a few hundred GB. However, graphs in real-world applications can easily exceed this size, such as web graphs and the Semantic Web. In fact, the file size of {\em ClueWeb}, a web graph used in our experiments, already exceeds 400GB, let alone those web graphs maintained by existing search engine companies, as well as other large graphs from online social networks and telecom operators.

One may, of course, increase the memory space in a cluster by adding more machines. However, since there are $\binom{n}{2}$ communication pairs in a cluster of $n$ machines and they all contend for the shared network resource, the communication overhead outweighs the increased computing power when $n$ becomes too large.

Due to the above reasons, researchers have recently developed out-of-core systems for processing big graphs. For example, Pregel-ix~\cite{pregelix} 
models the semantics of Pregel by relational operations like join and group-by, and leverages a general-purpose dataflow engine for out-of-core execution. Thus, it requires expensive external-memory join and group-by, which degrade the performance of distributed graph processing. Giraph also supports out-of-core execution, but \cite{pregelix} reported that it does not function properly.

Other out-of-core systems adopt the edge-centric Gather-Apply-Scatter model of PowerGraph~\cite{powergraph}, which is a special case of the vertex-centric model with a narrower application scope. Specifically, a vertex can only communicate with its adjacent vertex along an adjacent edge, which makes it unsuitable for algorithms that require pointer jumping (or path doubling)~\cite{ppa}. Moreover, the Gather-Apply phase is essentially message combining in Pregel. We further categorize these systems into two types as follows.

\vspace{2mm}

\noindent{\em Type~1: Single-Machine Systems.} Such systems include GraphChi~\cite{graphchi}, X-Stream~\cite{xstream}, and VENUS~\cite{venus}, which are designed to process a graph on a single PC. These systems require the IDs of vertices in a graph to be numbered as $1, 2, \cdots, |V|$, and vertices are partitioned into $P$ disjoint ID intervals, so that each partition can be loaded to memory for processing at a time. Besides the strict requirement on vertex ID format, these systems are also inefficient if only a small fraction of vertices need to perform computation in an iteration. This is because a whole partition needs to be loaded for processing as along as one vertex in it needs computation.

\vspace{2mm}

\noindent{\em Type~2: Distributed Systems.} We are only aware of one such system, Chaos~\cite{chaos}, which scales X-Stream out to multiple machines. However, Chaos represents the other extreme of hardware requirement with respect to single-PC systems. Specifically, its computation model is built upon the assumption that network bandwidth far outstrips storage bandwidth. In fact, \cite{chaos} reported that Chaos only achieves good performance by using large-SSD machines connected with 40 Gigabit Ethernet, but the performance is undesirable when Gigabit Ethernet is used, which is far more common in most small to medium size companies and most research institutes.

\vspace{2mm}

While there also exist some in-memory graph processing systems designed to run in a single big-memory machine, they are not designed to process very large graphs. For example, the largest graph tested with GRACE~\cite{grace} has less than 300 million edges. The high startup overhead is another problem, which we explain by considering the loading of a graph of 100GB size. In a distributed system running with 100 PCs, each PC only needs to load around 1GB data from HDFS (Hadoop Distributed File System). In contrast, a single-machine in-memory system needs to load all the 100GB
data from its local disk, and with its loading time alone a distributed system may have already finished many graph jobs.

In this paper, we introduce our GraphD system, which supports efficient out-of-core vertex-centric computation even with a small cluster of commodity PCs connected by Gigabit Ethernet, which is affordable to most users. We remark that GraphD aims at providing an efficient solution when memory space is insufficient; otherwise, one may simply use an in-memory system. GraphD specifically provides the following desirable features:

\vspace{2mm}

\noindent (1) When a graph $G=(V, E)$ is processed with $n$ machines using GraphD, we prove that each machine only requires $O(|V|/n)$ memory space, which allows GraphD to scale to very large graphs with a small cluster.

\vspace{2mm}

\noindent (2) By maintaining $O(|V|/n)$ vertex states in each machine, GraphD is able to automatically adapt the amount of edges streamed from local disks to the number of vertices that perform computation in an iteration, achieving high performance even in sparse computation workload (which is not possible in existing out-of-core systems).

\vspace{2mm}

\noindent (3) In a common cluster (connected with Gigabit Ethernet), network transmission is much slower than local disk streaming, and GraphD takes this insight into account in its design by a technique called ``outgoing message buffering'' to hide disk I/O cost inside network communication cost, leading to high performance as verified by our experiments.

\vspace{2mm}

\noindent (4) For a broad class of Pregel algorithms where message combiner is applicable, GraphD uses a novel ID-recoding technique to eliminate the need of any expensive external-memory join or group-by. In this case, the only external-memory operation is the streaming of edges and messages on local disks, GraphD is able to achieve almost the same performance as an in-memory Pregel-like system when disk I/O is fully hidden inside network communication.

\vspace{2mm}

The rest of this paper is organized as follows. Section~\ref{sec:related} reviews related work. Section~\ref{sec:dss} presents the distributed semi-streaming computation model of GraphD, and analyzes its space cost. Section~\ref{sec:basic} discusses the parallel framework of GraphD to fully overlap computation with communication. Section~\ref{sec:recode} describes the ID recoding technique and Section~\ref{sec:results} reports experimental results. Finally, Section~\ref{sec:conclude} concludes the paper.

\section{Background and Related Work}\label{sec:related}
We first review the computation model of Pregel, and then review other vertex-centric systems for processing big graphs. In this paper, we assume that the input graph $G=(V, E)$ is stored on HDFS, where each vertex $v\in V$ has a unique ID $id(v)$ and an adjacency list $\Gamma(v)$. For simplicity, we use $v$ and $id(v)$ interchangeably. If $G$ is undirected, $\Gamma(v)$ contains all $v$'s neighbors; while if $G$ is directed, $\Gamma(v)$ contains all $v$'s out-neighbors. The degree (or out-degree) of $v$ (i.e., $|\Gamma(v)|$) is denoted by $d(v)$. Each vertex $v$ also maintains a value $a(v)$ which gets updated during computation. A Pregel program is run on a cluster of machines, $\mathbb{W}$, deployed with HDFS.

\subsection{Pregel Review}\label{ssec:pregel}
\noindent{\bf Computation Model.} A Pregel program starts by loading an input graph from HDFS into memories of all machines, where each vertex $v$ is distributed to a machine $W=hash(v)$ along with $\Gamma(v)$, with $hash(.)$ being a partitioning function on vertex ID. Let $V(W)$ be the set of all vertices assigned to $W$. Each vertex $v$ also maintains a boolean field $active(v)$ indicating whether $v$ is active or halted.

A Pregel program proceeds in iterations, where an iteration is called a superstep. In Pregel, a user needs to specify a user-defined function (UDF) {\em compute}({\em msgs}) to be called by a vertex $v$, where {\em msgs} is the set of incoming messages received by $v$ (sent in the previous superstep). In $v$.{\em compute}(.), $v$ may update $a(v)$, send messages to other vertices, and vote to halt (i.e., deactivate itself). Only active vertices will call {\em compute}(.) in a superstep, but a halted vertex will be reactivated if it receives a message. The program terminates when all vertices are halted and there is no pending message for the next superstep. Finally, the results are dumped to HDFS.

To illustrate how to write {\em compute}(.), we consider the PageRank algorithm of~\cite{pregel} where $a(v)$ stores the PageRank value of vertex $v$, and $a(v)$ gets updated until convergence. In Step~1, each vertex $v$ initializes $a(v)=1/|V|$ and distributes $a(v)$ to its out-neighbors by sending each out-neighbor a message $a(v)/d(v)$. In Step~$i$ ($i > 1$), each vertex $v$ sums up the received message values, denoted by $sum$, and computes $a(v)=0.15/|V|+0.85\cdot sum$. It then distributes $a(v)/d(v)$ to each of its out-neighbors.

\vspace{2mm}

\noindent{\bf Combiner.} Users may also implement a message combiner to specify how to combine messages targeted at the same vertex $v_t$, so that messages generated on a machine $W$ towards $v_t$ will be combined into a single message by $W$ locally, and then sent to $v_t$. Message combiner effectively reduces the number of messages transmitted though the network. In the example of PageRank computation, the combiner can be implemented as computing sum, since only the sum of incoming messages is of interest in {\em compute}(.).

\vspace{2mm}

\noindent{\bf Aggregator.} Pregel also allows users to implement an aggregator for global communication. Each vertex can provide a value to an aggregator in {\em compute}(.) in a superstep. The system aggregates those values and makes the aggregated result available to all vertices in the next superstep.

\vspace{2mm}

For each vertex $v$, machine $W=hash(v)$ keeps the following information in main memory: (1)~the vertex state, which consists of $id(v)$, $a(v)$ and $active(v)$, and (2)~the adjacency list $\Gamma(v)$. Since vertex degree is required by out-of-core systems to demarcate adjacency lists of different vertices, to be consistent, we include $d(v)$ into the vertex state of $v$, which is given as follows:
\begin{equation}\label{eq_vstat}
state(v)=(id(v), a(v), active(v), d(v)).
\end{equation}

\subsection{Vertex-Centric Graph Processing Systems}\label{ssec:others}
As discussed in Section~\ref{sec:intro}, existing vertex-centric systems for big graph processing can be categorized into (1)~distributed in-memory systems, and (2)~single-PC out-of-core systems, and (3)~distributed out-of-core systems. We now review them in more detail.

\vspace{2mm}

\noindent{\bf Distributed In-Memory Systems.} Since Pregel~\cite{pregel} is only for internal use in Google, many open-source Pregel-like systems have been developed including Giraph~\cite{giraph}, GPS~\cite{gps}, GraphX~\cite{graphx}, and Pregel+~\cite{pregelplus}. Like Pregel, these systems keep an entire input graph in memories during computation, and also buffer intermediate messages in memories. While these systems adopt a synchronous execution model where vertex communicates by message passing, GraphLab~\cite{graphlab} adopts a different design. Specifically, a shared-memory abstraction is adopted where a vertex directly pulls data from its adjacent vertices/edges, and asynchronous execution is supported to allow faster convergence for algorithms where vertex values converge asymmetrically. A subsequent version of GraphLab, PowerGraph~\cite{powergraph}, partitions the graph by edges instead of vertices, in order to achieve more balanced workload. Since our work is more related to out-of-core systems, we refer interested readers to~\cite{tamerExp,ourExp} for more discussions on existing in-memory systems.

\vspace{2mm}

\noindent{\bf Single-PC Out-of-Core Systems.} These systems load one vertex partition to memory at a time for processing. In GraphChi~\cite{graphchi}, all vertices in a vertex partition and all their adjacent edges need to be loaded into memory before processing begins. X-Stream adopts a different design, which only needs to load all vertices in a partition into memory, while edges are streamed from local disk. Note that sequential streaming only requires a small in-memory buffer in order to achieve sequential I/O bandwidth, whose memory cost is negligible. In both GraphChi and X-Stream, a vertex communicates with each other by writing/reading data on adjacent edges. VENUS~\cite{venus} avoids the cost of writing data to edges, by letting a vertex obtain values directly from its in-neighbors, but it is not open source. However, all these systems need to scan the whole disk-resident input graph in each iteration, leading to undesirable performance for sparse computation workload. We remark that although GraphChi supports selective scheduling, it is ineffective since a whole partition (including all adjacent edges) needs to be loaded even if just one vertex in the partition needs computation.

\vspace{2mm}

\noindent{\bf Distributed Out-of-Core Systems.} Compared with single-machine systems, these systems only require a machine to process a partition of the graph, and thus the disk bandwidth of all machines are utilized in parallel. HaLoop~\cite{haloop} improves the performance of Hadoop for iterative computation by allowing a job to cache data to local disks to avoid remote reads in each iteration, but for vertex-centric graph computation, users need to explicitly program the interaction between vertices and messages using the MapReduce model. Pregelix~\cite{pregelix} formulates the computation model of Pregel using relational operations like join and group-by, and requires expensive external-memory join and group-by operations. Chaos~\cite{chaos} scales out X-Stream by partitioning the input graph on the disks of multiple machines, each of which streams its own portion of edges but may steal workload from other machines when it becomes idle. However, the system requires high-speed network to synchronize vertex values and to steal workloads, and is inefficient when Gigabit Ethernet is used.

\section{Data Organization and Streams}\label{sec:dss}
In this section, we describe the distributed semi-streaming (DSS) model of GraphD, analyze its memory cost and introduce its disk stream designs.

\subsection{Distributed Semi-Streaming Model}\label{ssec:mem}
We first consider the memory requirement of Pregel. For ease of analysis, we assume that the types of vertex ID, vertex value, adjacency list item, and message all have constant size. Accordingly, a vertex state as given in Eq~(\ref{eq_vstat}) also has constant size (as $active(v)$ and $d(v)$ have constant size). We remark that these data types are specified by users through C++ template arguments, and can have variable sizes in reality (e.g., vertex ID can be a string).

Recall that Pregel keeps the $O(|V|)$ vertex states, $O(|E|)$ edges (i.e., adjacency list items) in memories. Let us denote the set of messages currently in the system by $M$, where a messages is either on the sender-side or on the receiver-side. Then, $O(|M|)$ memory space is also required for keeping messages. Therefore, the total memory space required by Pregel is $O(|V|+|E|+|M|)$.

Note that $O(|E|)$ is typically much larger than $O(|V|)$. For example, a user in a social network can easily have tens of friends.

In many Pregel algorithms such as PageRank computation, only one message is transmitted along each edge in a superstep, and thus $O(|M|)=O(|E|)$. However, $|M|$ can be much larger in some Pregel algorithms. For example, in the triangle finding algorithm of~\cite{socialnet}, to confirm a triangle $\triangle v_1v_2v_3$ where $v_1<v_2<v_3$, $v_1$ needs to send $v_2$ a message asking about whether $v_3\in\Gamma(v_2)$ (note that $v_1$ has access to $v_2$ and $v_3$ in $\Gamma(v_1)$). Since there are $O(|E|^{1.5})$ triangles in a graph~\cite{triangle}, $O(|M|)$ is at least $O(|E|^{1.5})$.

According to the above analysis, the dominating memory cost is contributed by adjacency lists ($O(|E|)$) and messages ($O(|M|)$). GraphD streams adjacency lists and messages on local disks, leaving only the $O(|V|)$ vertex states in memories, and thus significantly reduces the memory requirement.

However, the $O(|V|)$ vertex states can still be too large to fit in the memory of a single machine. In fact, if all vertex states can fit in memory, single-PC systems such as GraphChi often provides an alternative model for more efficient semi-streaming graph processing. Since GraphD is a distributed system, each machine only needs to keep a portion of vertex states. GraphD follows the {\em distributed semi-streaming} (DSS) model\footnote{We name the model as DSS due to its similarity to semi-streaming computation of external-memory graph algorithms.}, where each machine $W$ only keeps the states of all vertices in $V(W)$ in its memory, and treats their adjacency lists and incoming and outgoing messages as local disk streams. It remains to show that DSS distributes the vertex states evenly among the $|\mathbb{W}|$ machines, i.e., each machine holds no more than $O(|V|/|\mathbb{W}|)$ vertex states with a small constant (e.g., 2). This memory requirement is very reasonable given the RAM size of a commodity PC today, allowing a small cluster to scale to very large graphs. We now prove this property below, where we regard the machine number $|\mathbb{W}|$ as a constant.

\begin{lemma}\label{lemma_space}
Assume that $hash(.)$ is well chosen so that a vertex is assigned to every machine with equal probability, then with probability of at least $(1 - O(1/|V|))$, it holds that $\max_{W\in\mathbb{W}}|V(W)|$ is less than $2|V|/|\mathbb{W}|$.
\end{lemma}
\begin{proof}
First, consider a particular machine $W$. Since every vertex is hashed to $W$ with probability $p=1/|\mathbb{W}|$, the total number of vertices that are hashed to $W$ (i.e., $|V(W)|$) conforms to a binomial distribution with mean $\mu=|V|p$ and variance $\sigma^2=|V|p(1-p)<|V|p=\mu$.

According to Chebyshev's inequality, we have
\begin{displaymath}
\Pr\Big(\Big||V(W)|-\mu\Big|\geq\mu\Big) \leq \sigma^2/\mu^2.
\end{displaymath}
Since $\sigma^2<\mu$, the R.H.S.\ is less than $1/\mu$. Moreover, since $|V(W)|$ is positive, the L.H.S.\ is equivalent to $\Pr(|V(W)|\geq2\mu)$. Therefore, we obtain
\begin{equation}\label{eq_1w}
\Pr(|V(W)|\geq 2\mu) < 1/\mu.
\end{equation}
Since $\mu=|V|/|\mathbb{W}|$, $1/\mu=|\mathbb{W}|/|V|=O(1/|V|)$ is a very small number. For example, when we process a billion-node graph using a cluster of 20 PCs, $|\mathbb{W}|$ is only 20 but $|V|$ is on the order of $10^9$, and thus $1/\mu$ is in the order of $10^{-7}$--$10^{-8}$.

We now consider all machines in $\mathbb{W}$, and proceed to prove our lemma:
\begin{eqnarray*}
& & \Pr(\max_{W\in\mathbb{W}}|V(W)| < 2|V|/|\mathbb{W}|)\\
& = & \Pr(\max_{W\in\mathbb{W}}|V(W)| < 2\mu)\\
& = & \Pr\big(\bigwedge_{W\in\mathbb{W}}\Big\{|V(W)| < 2\mu\Big\}\big)\\
& \geq & 1-\sum_{W\in\mathbb{W}}\Pr(|V(W)| \geq 2\mu)\ \ \ \ \ \mbox{(using union bound)}\\
& > & 1 - |\mathbb{W}|/\mu\ \ \ \ \ \ \ \ \ \ \ \ \ \ \ \ \ \ \ \ \ \ \ \ \ \ \ \ \ \ \,\mbox{(using Eq~(\ref{eq_1w}))}.
\end{eqnarray*}
The lemma is proved by noticing that $|\mathbb{W}|/\mu = |\mathbb{W}|^2/|V| = O(1/|V|)$. For example, when $|\mathbb{W}|$ is 20 and $|V|$ is in the order of $10^9$, $|\mathbb{W}|^2/|V|$ is in the order of $10^{-6}$--$10^{-7}$.
\end{proof}

We additionally require that main memory of a machine be large enough to hold the state $state(v)$ and adjacency list $\Gamma(v)$ of any single vertex $v$, so that $v$ can access them in $v.${\em compute}(.). We add this constraint because $\Gamma(v)$ of a high-degree vertex $v$ could require more memory space than $O(|V|/|\mathbb{W}|)$ (i.e., the bound of Lemma~\ref{lemma_space}), but this constraint is reasonable given the RAM size of a commodity PC today, and it is also required by existing out-of-core systems such as GraphChi and Pregelix.

\subsection{Graph Organization and Edge Streaming}\label{ssec:edgestream}
While GraphD may load data from HDFS and write results to HDFS, during the iterative computation, GraphD only sequentially reads and/or writes binary streams on local disks for efficiency.

When users specify GraphD to load an input graph from HDFS, the graph gets partitioned among all machines, where each machine $W$ saves the adjacency lists of vertices in $V(W)$ to local disk as an edge stream denoted by $S^E$. Meanwhile, the states of vertices in $V(W)$ are kept in memory (for computation) and also written to local disk (for subsequent local loading, see below). Optionally, if the graph is previously loaded from HDFS by another job, users may also specify GraphD to load graph from local disks, in which case each machine directly loads the previously saved vertex states to memory.

\begin{figure}[t]
    \centering
    \includegraphics[width=0.9\columnwidth]{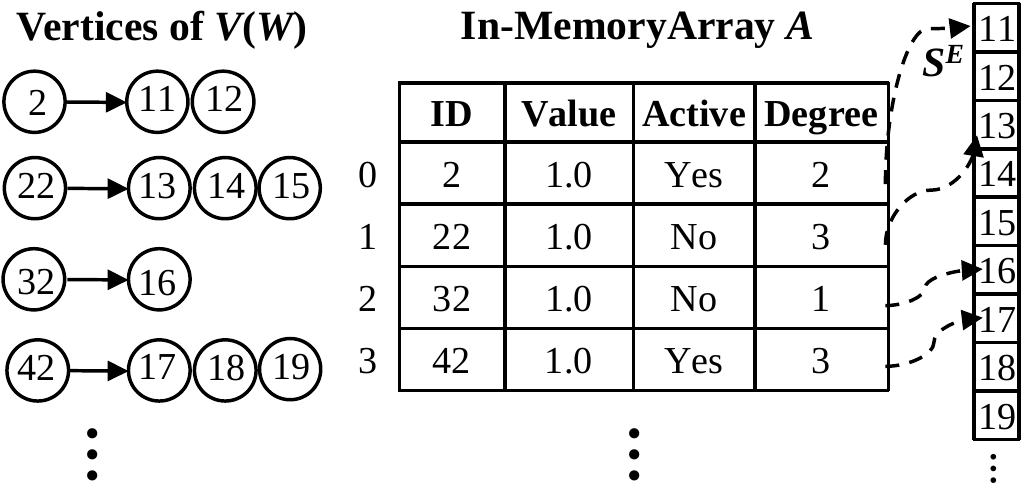}
    \caption{Vertex States and Edge Stream of a Machine $W$}\label{graphDS}
\end{figure}

In GraphD, each machine organizes its in-memory vertex states with an array $A$, as illustrated in Figure~\ref{graphDS}. Vertices in $A$ are ordered by vertex ID (i.e., 2, 22, 32, 42, $\cdots$ in Figure~\ref{graphDS}), and the edge stream $S^{E}$ simply concatenates their adjacency lists in the same order. In a superstep, {\em compute}(.) is scheduled to be called on the active vertices in $A$ in order. Since a vertex $v$ needs to access $\Gamma(v)$ in $v.${\em compute}(.), the next $d(v)$ items are sequentially read from $S^{E}$ to form $\Gamma(v)$. Thus, each superstep only sequentially reads $S^{E}$ once. If topology mutation is enabled, each superstep (say, Step~$i$) should digest an old edge stream $S^{E}_{\langle i-1\rangle}$ and generate a new edge stream $S^{E}_{\langle i\rangle}$ (for use in Step~$(i+1)$), where the subscripts denote the corresponding superstep number.

However, this method streams the whole edge stream once in each superstep, even if only a few vertices are active. Note that sparse computation workload is not a problem for in-memory systems since the adjacency lists are stored in RAMs, but it often causes performance bottleneck for disk-based systems like X-Stream. For example, \cite{xstream} admitted that X-Stream is inefficient for {\em graphs whose structure requires a large number of iterations}, as each iteration has to stream all edges of a graph. For the example of Figure~\ref{graphDS}, since Vertices~22 and~32 are not active, if they also receive no message, then their edges can be skipped.

For this purpose, our streaming algorithm should support a function {\em skip}({\em num\_items}), to skip the next {\em num\_items} items from the stream. Referring to Figure~\ref{graphDS} again, after Vertex~2 is processed, we may skip the edges of Vertices~22 and~32 by calling {\em skip}(4), where 4 is computed by adding their degrees $d(v)$ (i.e., 3 and 1 in array $A$). However, it is inefficient to perform a random disk read each time {\em skip}(.) is called. This is because, if there are many small series of inactive vertices in $A$, too many random disk I/Os are incurred, which may be even more costly than streaming the whole $S^E$.

We want our streaming algorithm to automatically adapt to the fraction of active vertices, i.e., (1)~it should achieve sequential disk bandwidth when the workload is dense, and (2)~should be able to skip a large amount of inactive vertices with a few random reads when the workload is sparse. We also need to guarantee that (3)~the worst case cost is no larger than streaming the whole $S^E$ once.

Before describing our streaming algorithm, we first consider how a stream (i.e., a file) is normally read. Specifically, an in-memory buffer $B$ of size $b$ is maintained throughout the streaming of a file. To read data from the stream, we continue reading from the latest read-position in $B$, and if we reach the end of $B$, we refill $B$ with the next $b$ bytes of data from the stream file on disk. Since each batch of $b$ bytes of data is read into $B$ using one random disk read, the costs of the random read (e.g., seek time and rotational latency) is amortized by all $b$ bytes, and as long as $b$ is not too small, the disk reads become sequential. GraphD sets $b$ to 64~KB as default, which is more than enough for achieving sequential bandwidth in most platforms, but is negligible given the RAM size of a modern PC.

To achieve the aforementioned 3 requirements, {\em skip}(.) avoids reading data from the file if after the skipping, the position to read data from is still in the buffer $B$. Obviously, this approach limits the number of random reads to be no more than that incurred when streaming the whole $S^E$. More specifically, {\em skip}($k$) is implemented as follows: we move the latest read-position in buffer $B$ forward for $k$ adjacency list items, to the position $pos$. If $pos$ is still inside $B$, we are done and no random read is incurred. Otherwise, $pos$ has exceeded the end of $B$, and we move the read-position in the stream file forward for $(pos-b)$ bytes (i.e., the amount to skip right after the end of $B$), to locate the start position for reading the next $b$ bytes from the file; we then refill $B$ with the next $b$ bytes of data read from the file. When the workload is sparse, this method is able to avoid sequentially reading a lot of useless items (by one random disk read).

Note that an important reason of maintaining vertex states in main memory is to quickly access vertex degrees for computing the number of bytes to skip. Otherwise, the vertex states can be treated as a disk stream that gets streamed along with $S^{E}$ during computation, in which case we only require a minimal amount of memory.

\subsection{Message Streams}

\begin{figure}[t]
    \centering
    \includegraphics[width=\columnwidth]{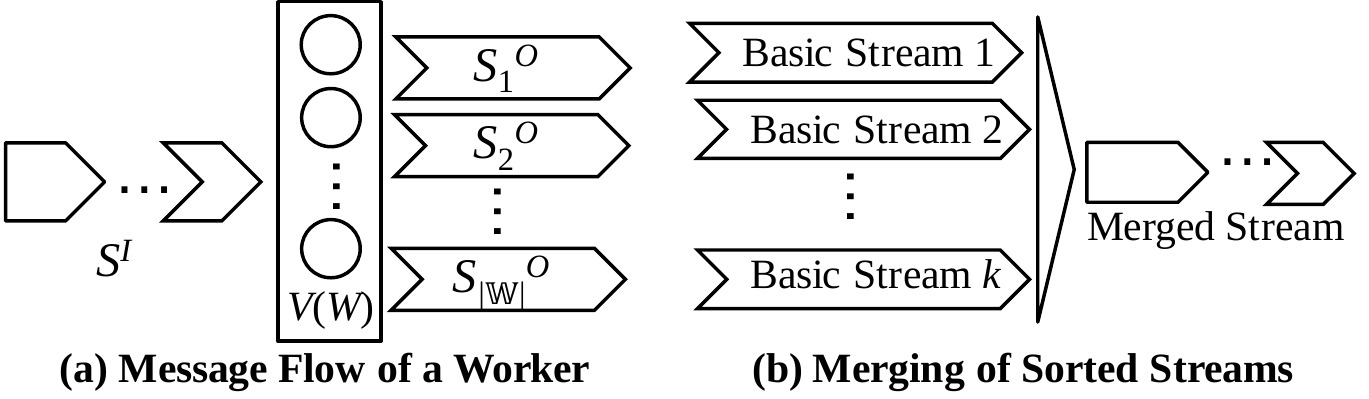}
    \caption{Message Streams in GraphD}\label{mstream}
\end{figure}

We now consider the message streams in our DSS model. As Figure~\ref{mstream}(a) shows, each machine can have an {\em incoming message stream} (\textbf{IMS}) $S^I$, and $|\mathbb{W}|$ {\em outgoing message streams} (\textbf{OMSs}) $S^O_i$ $(i=1, 2, \cdots, |\mathbb{W}|)$ on disk. Here, $S^O_i$ is used to buffer those messages towards vertices on the $i$-th machine, denoted by $W_i$. We now describe these streams.

\subsubsection{Outgoing Message Streams}\label{ssec:oms}
When a vertex $v$ sends a message to another vertex $u$ in $v$.{\em compute}(.), we may either (1)~buffer it in memory for sending, or we may (2)~append it to an OMS $S^O_{hash(u)}$ on local disk, to be loaded later for sending. The actual decision depends on whether disk streaming bandwidth or network bandwidth is larger. Option~(2) appears a bit strange at first glance, since it needs to write each message to disk and then load it back, leading to additional disk I/O. However, we need it since messages are generated by vertex-centric computation quickly, and if the speed of message sending cannot catch up but we keep buffering new messages, we might end up buffering too many messages in memory (or even causing memory overflow), breaching the bound of $O(|V|/|\mathbb{W}|)$ established in Lemma~\ref{lemma_space}. We explain when and why we adopt option~(1) or option~(2) below.

\vspace{2mm}

\noindent{\bf Design Philosophy.} If {\em network bandwidth is higher than disk streaming bandwidth} (e.g., when the 40 Gigabit Ethernet used in~\cite{chaos} is adopted), GraphD does not create OMSs. However, since vertex-centric computation generates messages quickly, the memory budget for buffering messages will soon be reached. To avoid memory overflow, vertex-centric computation is stalled while the buffered messages are being sent. Then, these messages are removed from the in-memory message buffer, allowing vertex-centric computation to continue (to generate and buffer more messages).

The stalling degrades performance, since it leads to repeated serial execution of message sending followed by vertex-centric computation (i.e., buffer refilling), and thus computation and communication do not overlap with each other. Note that the cost of computation in GraphD is not negligible since a machine needs to stream $S^E$. However, it makes no sense to write messages to OMSs, since writing a message to disk is even slower than sending it.

In contrast, if {\em disk streaming bandwidth is higher than network bandwidth} (e.g., when the commonly used Gigabit Ethernet is adopt-ed), GraphD 
uses OMSs to buffer the generated messages that are to be sent out. Since messages are buffered to local disks, vertex-centric computation is never stalled, allowing computation to be perform in parallel with message sending. The resulting parallelism of OMS appending and message sending, in turn, hides the cost of the former inside that of the latter, as disk bandwidth is higher than network bandwidth.

In fact, in a cluster assembled with PCs and 1 Gbps switches that are commonly available to average users, local disk streaming bandwidth is much larger than network bandwidth. This observation has been reported by existing work such as~\cite{ftgiraph}, which proposes a faster fault-recovery method for the framework of Pregel, but requires every machine $W$ to log all messages sent by $W$ also to the local disk. Experiments of~\cite{ftgiraph} (using 1 Gbps switch) reported that execution with message logging is almost as fast as execution without logging, which shows that the cost of sequentially writing messages to disks is negligible compared with message transmission. This is also confirmed by our experimental results reported in Table~\ref{pagerank_comp} of Section~\ref{sec:results}, which shows that the total time of vertex-centric computation (which performs message streaming) accounts for a very small fraction of the running time of a superstep, while message transmission lasts throughout the whole superstep.

We studied the reasons behind this observation, and found that (i)~disk streaming is significantly accelerated by OS memory cache, and that (ii)~the network resource is contended by all the $|\mathbb{W}|$ machines, limiting the connection throughput between any pair of machines. In this common setting, our use of OMSs is able to hide the disk I/O cost inside the communication cost, leading to full overlapping between computation and communication.

\vspace{2mm}

\noindent{\bf OMS Structure.} Recall that vertex-centric computation appends messages to an OMS $S^O_i$, and meanwhile, earlier messages in $S^O_i$ are loaded to memory for sending, and should then be garbage collected from $S^O_i$. One may organize an OMS as an append-only streaming file, where new data are always written to its in-memory buffer $B$ (of size $b=$ 64 KB), and when $B$ becomes full, the data in $B$ gets flushed to the stream file and $B$ is emptied for appending more messages.

However, this solution has several weaknesses. Firstly, since vertex-centric computation continually appends data to the OMS file, it is difficult to track whether sufficient new messages have been written to the file so that sending them will not underutilize the network bandwidth. Secondly, a message that gets sent cannot be garbage collected from its OMS. In short, it is not desirable to obtain messages from a file that is appending new data.

To solve this problem, we implement an OMS as a {\em splittable stream} that supports concurrent data appending (at the tail) and data fetching (at the head). Specifically, a splittable stream $S$ breaks a long stream of data items into multiple files $F_1, F_2, \ldots$, where each file $F_j$ either has at most $\mathcal{B}$ bytes, or contains only one data item whose size is larger than $\mathcal{B}$. Here, $\mathcal{B}$ is a parameter of a splittable stream and we shall discuss how to set it shortly.

A splittable stream $S$ appends data items to each of its file in a streaming manner, which only requires an in-memory buffer $B$ to achieve sequential disk I/O. Let us assume that $S$ is currently writing $F_j$. To append a data item $o$ to $S$, $S$ checks whether $F_j$ will have more than $\mathcal{B}$ bytes after appending $o$: (1)~if so, $F_j$ is closed and a new file $F_{j+1}$ is created for appending $o$; (2)~otherwise, $o$ is directly appended to $F_j$. Since $S$ writes to only one file at a time, $S$ requires only $b=64$~KB of memory.

In GraphD, since each OMS is organized as a splittable stream, the $|\mathbb{W}|$ OMSs in a machines take $|\mathbb{W}|\cdot b$ bytes of memory in total. Even when $|\mathbb{W}|=1000$, all OMSs take merely 64 MB of RAM.

\vspace{2mm}

\begin{figure}[t]
    \centering
    \includegraphics[width=\columnwidth]{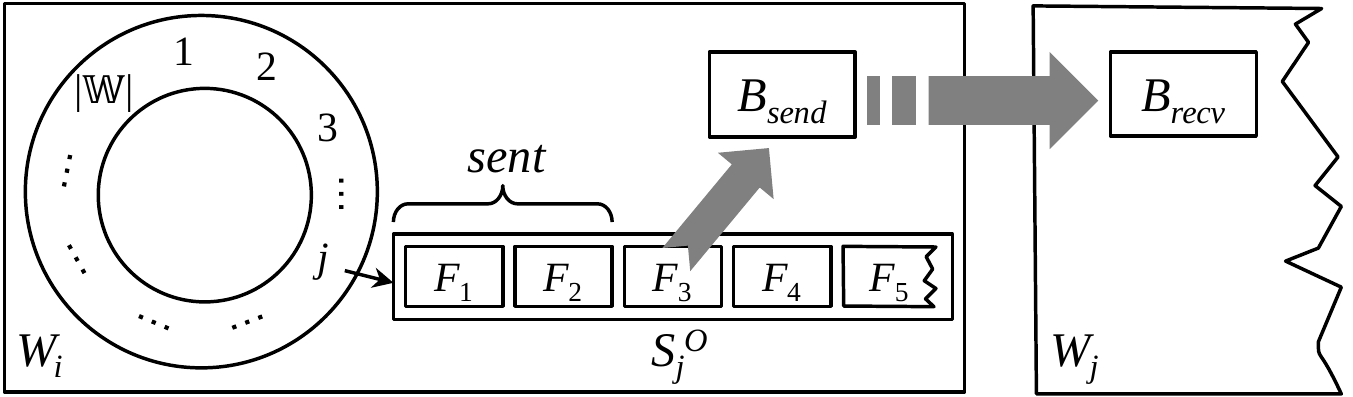}
    \caption{Sending Messages in OMSs}\label{fselect}
\end{figure}

\noindent{\bf Sending Messages in OMSs.} When an OMS $S^O_i$ is writing $F_j$, messages in $F_1, \ldots, F_{j-1}$ can be sent to machine $W_i$ to utilize the network bandwidth. We now describe how GraphD sends messages in OMSs. As Figure~\ref{fselect} shows, each machine $W_i$ maintains an in-memory sending buffer $B_{send}$, and a fully-written file $F_k$ of an OMS $S^O_j$ is sent to $W_j$ by first loading the messages in $F_k$ to $B_{send}$, which are then sent to $W_j$ in one batch. Obviously, the buffer size $|B_{send}|$ should be as large as the largest possible size of $F_j$, which is at least $\mathcal{B}$.

We will discuss how to set $|B_{send}|$ properly shortly. Now, we first discuss how we set $\mathcal{B}$. Obviously, the smaller $\mathcal{B}$ is, the finer-grained each message file is, and thus the less likely that message sending will be stalled on a file that is being appended. However, since messages are sent in batches of size around $\mathcal{B}$, $\mathcal{B}$ cannot be too small as sending messages in small batches is inefficient. GraphD sets $\mathcal{B}$ as 8 MB by default, which is large enough to fully utilize the network bandwidth (and keep the number of files tractable), while small enough to avoid file collision for both message appending and message sending. We remark that maintaining a sending buffer of 8 MB (or larger, as we shall explain) is well acceptable given the RAM size of a modern PC.

\vspace{2mm}

\noindent{\bf Sending Strategies.} Referring to Figure~\ref{fselect} again, each machine $W_i$ orders the $|\mathbb{W}|$ OMSs into a ring, where each OMS keeps track of the batch number of the last file that has been sent (resp.\ fully written), denoted by $no_s$ (resp.\ $no_w$). For example, for $S^O_j$ in Figure~\ref{fselect} which is currently appending messages to $F_5$, $no_s=2$ and $no_w=4$. Moreover, each machine keeps track of the position in the ring, denoted by $p$, from whose OMS (i.e., $S^O_p$) the previous message file is selected to be loaded to $B_{send}$ for sending.

If message combiner is not used, we scan through the ring from position $p$, until an OMS $S^O_j$ is reached whose $no_s<no_w$ (i.e., there is at least one file to send). There are two possible cases.

{\em Case 1:} if such an OMS is found before the scan reaches $p$ again, we load $F_{no_s+1}$ to $B_{send}$ for sending, and then update $p$ as $j$. For example, for $S^O_j$ in Figure~\ref{fselect}, we only send $F_3$. Then, the same scan operation is repeated starting from the updated position $p$ in the ring. Note that even if $S^O_j$ has more than one file to send to $W_j$ (e.g., $F_4$ in Figure~\ref{fselect}), the next scan will pick a file from another OMS $S^O_{j'}\ (j'\neq j)$ for sending to $W_{j'}$ (if exists), to avoid communication bottleneck on the receiver-side. For the same reason, different machines will initialize $p$ to be different values when a job begins.

{\em Case 2:} if the scan reaches $p$ again without finding a valid OMS, then no OMS has a file to send, and thus the scanning thread goes to sleep. The thread is awakened to repeat the scan whenever a new message file is written.

On the other hand, if message combiner is used, we adopt a different scanning strategy to maximize the effect of message combining: if the scan locates a valid OMS, all its message files from $F_{no_s+1}$ to $F_{no_w}$ are combined for sending in one batch. Specifically, the messages are first merge-sorted (i.e., grouped) by destination vertex ID, and then another pass over the sorted messages combines each group into one message and appends this message to $B_{send}$ for sending. The strategy is effective, since (1)~when all active vertices have called {\em compute}(.) in the current superstep, OMSs are finalized and our strategy essentially combines all remaining messages in each OMS, while (2)~otherwise, message combining runs in parallel with vertex-centric computation, and thus does not increase the computation time.

Here, combined messages are appended to $B_{send}$, and since the messages come from multiple files, $|B_{send}|$ may need to increase beyond $\mathcal{B}$. However, since there could be at most one combined message for each vertex in the target machine, $|B_{send}|$ is upper bounded by $O(\max_{W\in\mathbb{W}}|V(W)|)$. While GraphD sets $|B_{send}|$ as $\mathcal{B}$ by default, if combiner is used, GraphD increases $|B_{send}|$ to $O(\max_{W\in\mathbb{W}}|V(W)|)$ (if it is larger than $\mathcal{B}$). According to Lemma~\ref{lemma_space} of Section~\ref{ssec:mem}, the memory bound of $O(|V|/|\mathbb{W}|)$ is still kept.

Finally, we show that merge-sorting message files takes only constant memory space which is well-affordable to a modern PC. Assume that we sort files $F_1, F_2, \ldots, F_n$ by $k$-way merge-sort, then it takes $\lceil\log_{k}n\rceil$ sequential passes over all the messages. At any time during the merge-sort, only one merge operation is running where (at most) $k$ sorted message files are being merged into one larger message file as Figure~\ref{mstream}(b) illustrates. Since we treat each sorted message file as a stream when reading/appending messages, the merge-sort uses $(k+1)$ in-memory buffers, which takes $(k+1)b$ memory space.

GraphD sets $k$ to 1000, and thus a merge-sort operation takes merely $(64$ MB $+$ $64$ KB$)$ RAM despite the large $k$. Moreover, the large value of $k$ allows merge-sort to take only one pass even for very large graphs, since the number of message files to combine is usually smaller than $k=1000$. To see this, recall that each message file has size around $\mathcal{B}=8$ MB, and thus $k$ files have size around 8 GB, which is quite large for an OMS (which only contains messages transmitted between one pair of machines).

Also note that this strategy is just a baseline, and in Section~\ref{sec:recode} we shall see that when our ID recoding technique is used, messages can be combined in memory without performing merge-sort first.

\subsubsection{Incoming Message Stream}
We now consider the IMS $S^I$. Since outgoing messages are loaded to $B_{send}$ and sent in batches, each machine also needs to maintain an in-memory receiving buffer $B_{recv}$ with $|B_{recv}|=|B_{send}|$. In each machine, a receiving thread listens on the network, and uses $B_{recv}$ to receive one message batch at a time and adds the messages to $S^I$. Next, we discuss how to add received messages to $S^I$.

In a superstep, each active vertex $v$ calls {\em compute}({\em msgs}), where {\em msgs} is obtained from $S^I$. Since the vertex-state array $A$ and edge stream $S^E$ are already ordered by vertex ID, we require messages in $S^I$ also to be ordered by destination vertex ID, so that vertex-centric computation may simply proceed in one pass over $A$ by sequentially reading from both $S^I$ and $S^E$. Specifically, to call $v$.{\em compute}({\em msg}), $v$ may read the next $d(v)$ items from $S^E$, and sequentially read messages targeted at $v$ from $S^I$ and append them to {\em msgs} until a message targeted at $u>v$ (or the end of $S^I$) is reached.

However, the order that messages in $S^I$ are received depends on the actual communication process. We adopt the following approach to make $S^I$ ordered. Specifically, whenever a machine receives a batch of messages in $B_{recv}$, it sorts the messages by destination vertex ID, and then writes the sorted messages to a file on disk. Finally, when all incoming messages for the current superstep are received, the sorted message files are further merged (or merged-sorted) into one sorted message file, which becomes $S^I$.

As we have discussed, GraphD uses 1000-way merge-sort which takes merely $(64$ MB $+$ $64$ KB$)$ RAM. Moreover, since each received message batch has size around 8 MB, when there are no more than 8 GB messages, the message files are simply merged. Moreover, merge-sort is unlikely to take more than 2 passes since this requires a machine to receive over 8 TB messages.

Again, this solution is just a baseline, as the use of the ID recoding technique allows incoming messages to be digested in memory, and thus there is no $S^I$ (and no merge-sort).

\subsubsection{Cost Analysis}\label{ssec:cost}
We now analyze the total memory cost of GraphD when both IMS and OMSs are maintained, assuming that $|\mathbb{W}|<1000$.

For communication, each machine maintains two buffers $B_{send}$ and $B_{recv}$ which take $2\mathcal{B}=16$ MB memory space. For computation, the OMSs need $|\mathbb{W}|\cdot b<64$ MB memory for appending messages, and $2b=128$ KB memory for reading the edge stream and the IMS. When combiner is used, the merge-sort for combining messages before sending takes $(64$ MB $+$ $64$ KB$)$ memory. After all messages are received, the merge-sort for constructing $S^I$ needs $(64$ MB $+$ $64$ KB$)$ memory. Therefore, each machine requires only around 200 MB additional memory space besides the vertex-state array $A$. Therefore, the space bound of $O(|V|/|\mathbb{W}|)$ established by Lemma~\ref{lemma_space} is still kept.

As for the disk I/O cost of a superstep, all the streams $S^E$, $S^I$ and $S^O_i$ are sequentially read and/or written for only one pass, while the merge-sort for combining messages (resp.\ for constructing $S^I$) additionally takes one pass (or at most two passes for an enormous graph) over the outgoing (resp.\ incoming) messages.

\subsection{Other Issues}\label{ssec:issues}
\noindent{\bf Data Loading.} Data loading from HDFS is similarly processed, except that now data items in an OMS and an IMS becomes vertices (along with their adjacency lists) rather than messages.

Specifically, each machine parses a portion of the input file from HDFS, and for each vertex $v$ parsed, it is appended to the OMS $S^O_i\ (i=hash(v))$ to be directed to $W_i$. Since there could be high-degree vertices, we require $|B_{send}|$ and $|B_{recv}|$ to be at least large enough to hold one vertex and its adjacency list (which could be larger than $O(|V|/|\mathbb{W}|)$) at the loading stage. The received vertices are merge-sorted by vertex ID into $S^I$, which then gets splitted into $A$ and $S^E$ using another pass over $S^I$.

\vspace{2mm}

\noindent{\bf Topology Mutation.} GraphD also supports algorithms that perform topology mutation, by associating a type with each message indicating whether the message is an ordinary one or is for vertex mutation. Edge mutations are performed in $v$.{\em compute}(.) by directly updating $\Gamma(v)$, which is written to a new local edge stream for the next superstep. Vertex mutations are performed after vertex-centric computation, where new vertices are appended to the vertex-state array $A$, and deleted vertices are simply marked in $A$ without actual deletion. This design guarantees that existing vertices never change their positions in $A$, which is required for our ID recoding technique to be described in Section~\ref{sec:recode}.

\vspace{2mm}

\noindent{\bf Fault Tolerance.} Our design of streams naturally supports checkpointing. The vertex states and edge streams are backed up to HDFS at the beginning of a job. To checkpoint a superstep, the IMSs of all machines are backed up to HDFS; and if topology mutation happens, the locally-logged incremental updates since last checkpoint are also backed up to HDFS. When failure happens, a machine loads its vertex states and edge stream from HDFS, replays the mutation operations, and loads incoming messages from the latest checkpoint to resume execution.

Our DSS model straightforwardly supports the message-log based fast recovery approach of~\cite{ftgiraph} mentioned in Section~\ref{ssec:oms}, since every machine writes outgoing messages to OMSs on local disk. The only change required is the timing of garbage collecting OMSs: each machine keeps all its OMSs on local disk (for use during recovery) until a new checkpoint is written to HDFS, instead of deleting a message file immediately after its messages are sent.

\section{Parallel Framework of DSS}\label{sec:basic}
In Section~\ref{sec:dss}, we have discussed the graph and stream data organization of our DSS model. In this section, we introduce how these structures are actually used in parallel graph computation. We focus on both the parallelism between machines, and that within a machine. Due to the space limitation, we only discuss the most complicated case when all streams $S^E$, $S^I$ and $S^O_i$ are used by GraphD, i.e., when local disk streaming bandwidth is larger than network bandwidth as is common in a commodity cluster. In this setting, the intra-machine parallelism mainly refers to the overlapping of computation (local disk streaming) with communication (message transmission). We now present our parallel framework.

We use {\em FIFO} communication channels, i.e., if a machine $W_i$ sends message $m_a$ and then $m_b$ to another machine $W_j$, $W_j$ is guaranteed to receive $m_a$ before receiving $m_b$. We also use {\em condition variables} to avoid a blocking thread from occupying CPU resources: suppose that a thread $w_a$ needs to block until a condition holds, it may wait on a condition variable {\em cond-var} and will no longer occupy CPU, when another thread $w_b$ updates the condition, it may wake up $w_a$ to continue execution using {\em cond-var}.

Each machine runs three units in parallel: (1)~a sending unit $U_s$ that sends outgoing messages; (2)~a receiving unit $U_r$ that receives incoming
messages; and (3)~a computing unit $U_c$ that performs vertex-centric computation (to generate messages). We now explain how they interact with each other.

\vspace{2mm}

\noindent{\bf Synchronization Between Supersteps.} Since Pregel adopts synchronous execution, it is unreasonable to delay the transmission of messages generated in Step~$i$, by transmitting messages generated in Step~$(i+1)$, especially in a commodity cluster where network bandwidth is limited.

Therefore, $U_s$ of all machines should block the sending of messages generated by their $U_c$ in Step~$(i+1)$, until all messages generated in Step~$i$ have been received by $U_r$ in all machines.

Our framework guarantees this property, by letting $U_r$ in each machine to synchronize with the receiving units of all other machines, after it has received all the messages towards its machine (generated in Step~$i$). We will discuss how $U_r$ determines this condition shortly. After the synchronization, $U_r$ guarantees that all messages generated in Step~$i$ have been transmitted, and thus it notifies $U_s$ to continue sending messages generated in Step~$(i+1)$.

\vspace{2mm}

\noindent{\bf Message Receiving.} We now explain how $U_r$ decides whether it has received all messages of Step~$i$. Specifically, whenever $U_s$ in a machine $W_j$ has sent all its messages towards another machine $W_k$ (i.e., $W_j$'s OMS $S^O_k$ is exhausted), it will send an end tag (a special message) to $W_k$. As a result, a machine $W_k$ just needs to count the number of end tags received, and if it reaches $|\mathbb{W}|$, messages from all machines must have been received. This is because GraphD guarantees the property that all messages (including end tags) generated in Step~$i$ must be transmitted before any message (including an end tag) generated in Step~$(i+1)$, as we have described before.

Here, $U_s$ decides that it has exhausted its OMS $S^O_k$ (and sends an end tag to $W_k$) if the following two conditions are met: (1)~$U_c$ has finished vertex-centric computation for Step~$i$, and will thus generate no more messages of Step~$i$; and (2)~there is no more message file in OMS $S^O_k$ for sending.

\vspace{2mm}

\noindent{\bf Vertex-Centric Computation.} When $U_c$ finishes its computation of Step~$i$, it has to block until $U_r$ has received all messages towards it in Step~$i$, before starting to compute Step~$(i+1)$. This is because, to call $v$.{\em compute}({\em msg}) in Step~$(i+1)$, we need to guarantee that {\em msg} contains all the messages sent to $v$ from Step~$i$.

However, unlike $U_s$, $U_c$ does not need to wait till all receiving units are synchronized, and may start generating messages of Step~$(i+1)$ earlier, although these messages will only be sent by $U_s$ after the synchronization.

To summarize, in Step~$i$, $U_r$ first keeps receiving messages until $|\mathbb{W}|$ end tags are received, then notifies $U_c$ that it is allowed to compute Step~$(i+1)$, then synchronizes with the receiving units of the other machines; and if the job should continue, $U_r$ then notifies $U_c$ that it is allowed to send messages for Step~$(i+1)$.

The benefit of letting $U_c$ start computing Step~$(i+1)$ earlier is that, when $U_s$ starts to send messages of Step~$(i+1)$, it can readily find fully-written OMS files for sending, and thus can fully utilize the network bandwidth.

\vspace{2mm}

\noindent{\bf Synchronization of Global Information.} When $U_c$ of a machine $W$ performs vertex-centric computation in Step~$i$, it will aggregate data to its local aggregator, and update local control information such as whether $W$ has sent any message and whether any vertex is active after calling {\em compute}(.). These data needs to be synchronize to decide whether to continue computing Step~$(i+1)$, and to obtain the global aggregator value for use by {\em compute}(.) in Step~$(i+1)$.

Although we can do that during the previously-described synchronization among the receiving units, this may delay the computation of Step~$(i+1)$ since $U_c$ needs to wait for the global aggregator value, which in turn needs to wait for the transmission of all messages generated in Step~$i$ (recall that we assume communication bandwidth is limited).

Instead, we let the computing units of all machines synchronize these global data as soon as they finish their vertex-centric computation, and there is no need to wait for the slower message transmission to complete. This allows $U_c$ to start computing a new superstep much earlier than the synchronization among receiving units. If $U_c$ decides that the job should terminate after synchronizing with other computing units, it signals $U_s$ and $U_r$ to terminate after they finish processing their current superstep, and then terminates itself.

\section{The ID-Recoding Technique}\label{sec:recode}
For Pregel algorithms where message combiner is applicable, GraphD supports a more efficient computation model which uses a novel ID-recoding technique to (1)~directly digest incoming messages in memory, which eliminates $S^I$, and to (2)~combine outgoing messages in memory, which eliminates the need of external-memory merge-sort on OMS files. Recall the disk I/O cost from Section~\ref{ssec:cost}, then the two aforementioned improvements essentially means that in a superstep, the only disk I/O cost is to streaming $S^E$, and to sequentially appending messages to OMSs. In other words, each superstep only requires one sequential pass over the edge stream, and one sequential pass over the generated messages, which is almost the {\bf minimum} possible I/O cost that any out-of-core Pregel-like system can achieve (if edges and messages are streamed on disks). In contrast, the state-of-the-art out-of-core system, Pregelix, still performs expensive external-memory sort and group-by operations even for algorithms where combiner applies.

We remark that this kind of algorithms cover a broad number of Pregel algorithms. In fact, many systems adopt a narrower edge-centric Gather-Apply-Scatter (GAS) computation model, such as PowerGraph~\cite{powergraph}, GraphChi~\cite{graphchi} and X-Stream~\cite{xstream}, and this model is essentially Pregel algorithms with message combining. Specifically, in the GAS model, each vertex $v$ gathers a message from every in-edge, and combines them to update the value of $v$.

\vspace{2mm}

\noindent{\bf Vertex ID Recoding.} Vertex ID recoding is required by many graph systems to enable efficient computation, although their purposes are different. For example, single-machine systems like Graph-Chi~\cite{graphchi}, X-Stream~\cite{xstream} and VENUS~\cite{graphchi} all require the 
vertex IDs to be numbered as $1, 2, \cdots$, since their computation model partitions vertices based on vertex ID intervals. While distributed Pregel-like systems should allow users to specify the type of vertex ID (e.g., using a C++ template argument), one such system, GPS~\cite{gps}, still requires the vertex IDs to be numbered as $1, 2, \cdots$. Since the vertex IDs are dense, when a message targets at vertex $v$ is received, the machine can directly locate the incoming message queue of $v$ to append the message, without looking up its location from a lookup table (which incurs much overhead since the lookup is needed for every message). As a result, GPS can achieve better performance than other systems like Giraph and GraphLab, as reported in~\cite{ourExp}. Giraph++~\cite{giraph++} partitions a graph to allow more efficient computation, but in order to allow the vertex-to-machine mapping to still be captured by a simple function $hash(.)$ during sebsequent computation, it also recodes the vertex IDs.

\begin{figure}[t]
    \centering
    \includegraphics[width=\columnwidth]{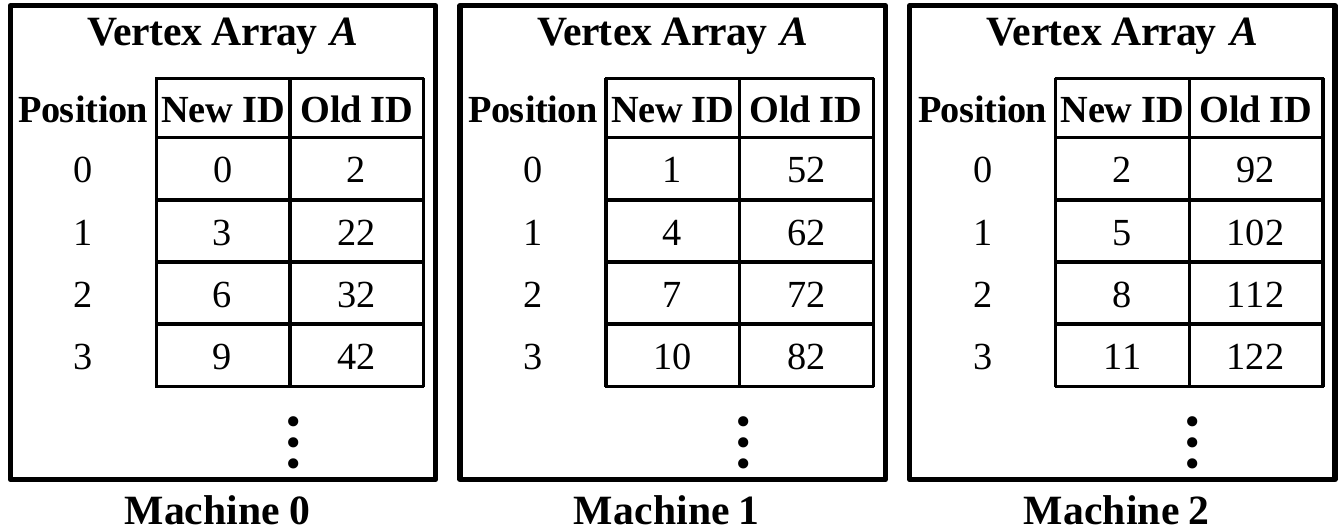}
    \caption{Example of ID Recoding}\label{recode}
\end{figure}

We design a new ID recoding for GraphD to allow in-memory message combining and digesting, while retaining the memory bound of $O(|V|/|\mathbb{W}|)$ established by Lemma~\ref{lemma_space}. The key idea is to establish a one-to-one mapping between the ID of a vertex and its position in the state array $A$, which is efficient to compute. We now explain how GraphD establishes this mapping, assuming that vertex IDs are numbered as $0, 1, \cdots, |V| - 1$. In GraphD, machines are numbered as $0, 1, \cdots, |\mathbb{W}| - 1$. When GraphD is running in {\em recoded mode}, it uses the vertex partitioning function $hash(v)=id(v)$ modulo $|\mathbb{W}|$.

As an illustration, Figure~\ref{recode} shows the vertex state arrays $A$ of a cluster of 3 machines processing a graph with 12 vertices, where we only show the old IDs and the new (i.e., recoded) IDs of the vertices in $A$. We can see that the old IDs are sparse, and are now recoded into dense IDs numbered as $0, 1, \cdots, 11$. In the recoded mode, the new IDs are treated as the actual vertex ID, and obviously we have $hash(v)=id(v)$ modulo 3 (e.g., Vertex~5 is assigned to Machine~2 since 5 modulo 3 is equal to 2).

For a vertex at position $pos$ of array $A$ in Machine~$i$, we can compute its new ID as $(|\mathbb{W}|\cdot pos+i)$. For example, in Figure~\ref{recode}, the vertex whose old ID is 102 is at position~1 of array $A$ in Machine~2, and thus its new ID is computed as $(3\cdot1+2)=5$. Moreover, given the new ID of a vertex, $id$, on Machine~$i$, we can compute its position in $A$ as $\lfloor id/|\mathbb{W}|\rfloor$. For example, in Figure~\ref{recode}, the vertex whose new ID is 5 (in Machine~2) is at position $\lfloor 5/3\rfloor=1$.

\vspace{2mm}

\noindent{\bf Preprocessing.} To run a job in recoded mode, either the vertices already have their IDs numbered as $0, 1, \cdots, |V| - 1$, or we need to preprocess the graph to assign its vertices with new IDs $0, 1, \cdots, |V| - 1$. We now describe our algorithm for the preprocessing, which is essentially a GraphD job running in normal mode (and thus requires only $O(|V|/|\mathbb{W}|)$ memory on each machine).

In preprocessing, the old IDs are used as the input to $hash(.)$ called during vertex assignment and message passing. After the input graph is loaded, each machine scans array $A$ and assigns each vertex a new ID which is computed from its position in $A$. However, for each vertex $v$, the neighbor IDs in $\Gamma(v)$ (stored in $S^E$) are still the old IDs, and we need to replace them with their new IDs (which are required for sending messages in recoded mode later).

For a directed graph, recoding the IDs in adjacency lists takes 3 supersteps. Let us denote the old (resp.\ new) ID of a vertex $v$ by $id_{old}(v)$ (resp.\ $id_{new}(v)$). In Step~1, each vertex $v$ sends $id_{old}(v)$ to every out-neighbor $u\in\Gamma(v)$ asking for $id_{new}(u)$. In Step~2, a vertex $u$ responds to each requester $id_{old}(v)$ (recall that $hash(.)$ takes the old ID) by sending it $id_{new}(u)$. Finally, in Step~3, each vertex $v$ simply appends the received new neighbor IDs to a new edge stream $S^E_{rec}$, which is treated as the edge stream for streaming in recoded mode later. For an undirected graph, we skip Step~1 since a vertex $u$ can directly send $id_{new}(u)$ to each neighbor $v\in\Gamma(u)$.

The whole recoding process sends only $O(|E|)$ messages, and our experimental results in Section~\ref{sec:results} show that the preprocessing time is comparable to that of parallel graph loading from HDFS. If graph recoding is performed right after we put $G$ onto HDFS, it adds very little additional time compared with the time for putting $G$.

\vspace{2mm}

\noindent{\bf Execution in Recoded Mode.} If the vertex IDs of the original graph are already numbered as $0, 1, \cdots, |V| - 1$, our recoded mode can directly load it from HDFS. Otherwise, we need to preprocess the input graph as described above. After the graph is recoded, state array $A$ and stream $S^E_{rec}$ of each machine are already on its local disks, and thus our recoded mode simply let each machine load $A$ to memory, and stream $S^E_{rec}$ on local disk (instead of $S^E$).

Our recoded mode additionally requires users to specify an identity element $e^0$, such that when we combine $e^0$ with any message $m$, the combined message is still $m$. For example, $e^0=0$ for PageRank computation since $e^0 + m = m$; while if the operation of the combiner is to take minimum, $e^0$ can be set as $\infty$.

\vspace{2mm}

\noindent{\em In-Memory Message Digesting.} In recoded mode, $U_r$ now directly digests messages in memory. Specifically, in Step~$i$, before receiving messages, $U_r$ first creates an in-memory array with $|V(W)|$ message elements, denoted by $A_r$. Here, $A_r[pos]$ refers the combined message towards the corresponding vertex of $A[pos]$.

Each element in $A_r$ is initialized as $e^0$. When a batch of messages is received into $B_{recv}$, for each message, we compute the position of its destination vertex $u$ in array $A$ from $u$'s ID, which is $pos = \lfloor id(u)/|\mathbb{W}|\rfloor$, and then combine the message to $A_r[pos]$.

After all messages generated in Step~$i$ are received and $U_c$ starts processing Step~$(i+1)$, the corresponding vertex of $A[pos]$ is regarded as having received messages only if $A_r[pos]\neq e^0$, in which case {\em compute}({\em msgs}) is called on the vertex with {\em msgs} containing only the combined message $A_r[pos]$. When $U_c$ finishes computing Step~$(i+1)$, it frees $A_r$ from memory.

Let us define $A^{(i)}_r$ as the array $A_r$ that is created by $U_r$ for receiving messages generated in Step~$(i-1)$ and then freed by $U_c$ after it finishes computing Step~$i$. Then, two arrays of $A_r$ coexist in any superstep: in Step~$i$, $U_r$ creates $A^{(i+1)}_r$ and updates it with received messages (for use by $U_c$ in Step~$(i+1)$), while $U_c$ obtains incoming messages from $A^{(i)}_r$ for computation. Since the two arrays require $O(|V(W)|)$ memory, according to Lemma~\ref{lemma_space}, the memory bound of $O(|V|/|\mathbb{W}|)$ still holds.

\vspace{2mm}

\noindent{\em In-Memory Message Combining.} Similarly, $U_s$ always maintains an in-memory array with $max_{W\in\mathbb{W}}|V(W)|$ message elements, denoted by $A_s$, for combining outgoing messages. According to Lemma~\ref{lemma_space}, maintaining $A_s$ does not breach the memory bound of $O(|V|/|\mathbb{W}|)$.

Each element of $A_s$ is initialized as $e^0$. Recall that $U_s$ combines and sends those messages from one OMS (i.e., towards one destination machine) at a time. To combine a set of messages towards machine $W_i$, for each message that targets at a vertex $u$, $U_s$ computes its position in array $A$ of the destination machine, which is $pos = \lfloor id(u)/|\mathbb{W}|\rfloor$, and then combines the message to $A_s[pos]$.

After all messages in an OMS are combined to $A_s$, for each message element $A_s[pos]\neq e^0$, $U_s$ attach the message value with the ID of its target vertex, which is $|\mathbb{W}|\cdot pos+i$; $U_s$ then appends the target-labeled message to $B_{send}$ for sending. To guarantee that all elements of $A_s$ are $e^0$ before combining the next batch of message files, $U_s$ also sets $A_s[pos]$ back to $e^0$ after the corresponding message gets appended to $B_{send}$.

\vspace{2mm}

\noindent{\em Topology Mutation.} Topology mutation is handled similarly as described in Section~\ref{ssec:issues}, with a change for vertex addition. Specifically, in a superstep, after vertex-centric computation, $U_c$ first recodes the IDs of the newly added vertices by synchronizing with the computing units of other machines, using the same method as in preprocessing; $U_c$ then appends these recoded vertices to $A$ (implemented using STL vector). The cost of the above intra-superstep id-recoding operation is proportional to the number of vertices added.

\section{Experiments}\label{sec:results}
We evaluate the performance of GraphD by comparing it with distributed out-of-core systems Pregelix (Release 0.2.12) and HaLoop, and single-PC out-of-core systems GraphChi and X-Stream (v1.0). We also report the performance of an in-memory system, Pregel+, as a reference to measure the disk I/O overhead incurred by out-of-core execution. Pregel+ is a fair choice since it has been shown to outperform other in-memory graph systems for various algorithm-graph combinations in a recent performance study~\cite{ourExp}. The source code of our GraphD system and all the applications used in our evaluation are available from: \url{http://www.cse.cuhk.edu.hk/systems/graphd}.

All experiments were conducted on two clusters, both connected by Gigabit Ethernet. The first cluster consists of 16 commodity PCs, each with four 3.40GHz processors (Intel Core i5-4670), 8GB RAM and a 320GB disk. The PCs are connected by an unmanaged switch that provides a relatively low network speed. The second cluster consists of 15 servers, each with twelve 2.0GHz cores (two Intel Xeon E5-2620 CPUs), 48GB RAM and a 200GB disk. In addition, one server additionally has access to another 2TB disk. These servers are connected by Cisco C2960 switch which provides a relatively high network speed.

We denote the first cluster by $\mathbb{W}^{PC}$ and the second one by $\mathbb{W}^{high}$. For distributed systems, all machines in a cluster were used; while for single-PC systems, only one of the machines was used. Notably, $\mathbb{W}^{high}$ has 0.72TB memory space in total, and we use it in order to compare the out-of-core systems with the in-memory Pregel+ system running with enough memory (as the memory space of $\mathbb{W}^{PC}$ is insufficient to run Pregel+ in most graphs we tested).

\begin{table}[!t]
\centering
\caption{Graph Datasets}\label{data}
\includegraphics[width=\columnwidth]{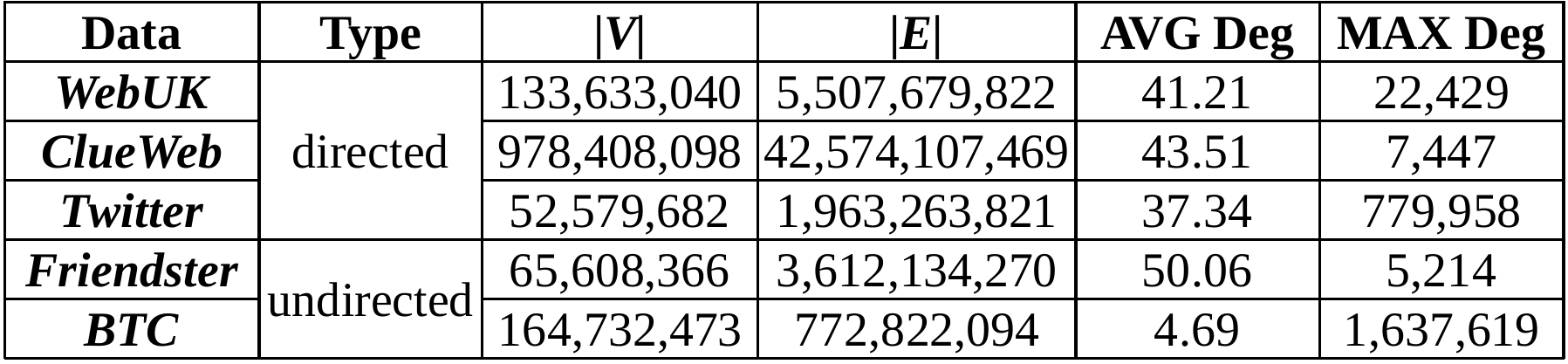}
\end{table}

\begin{table*}[t]
\centering
\caption{Performance of PageRank Computation on $\mathbb{W}^{PC}$ (time marked with $\star$: smallest among all systems)}\label{pagerank_pc}
\includegraphics[width=1.7\columnwidth]{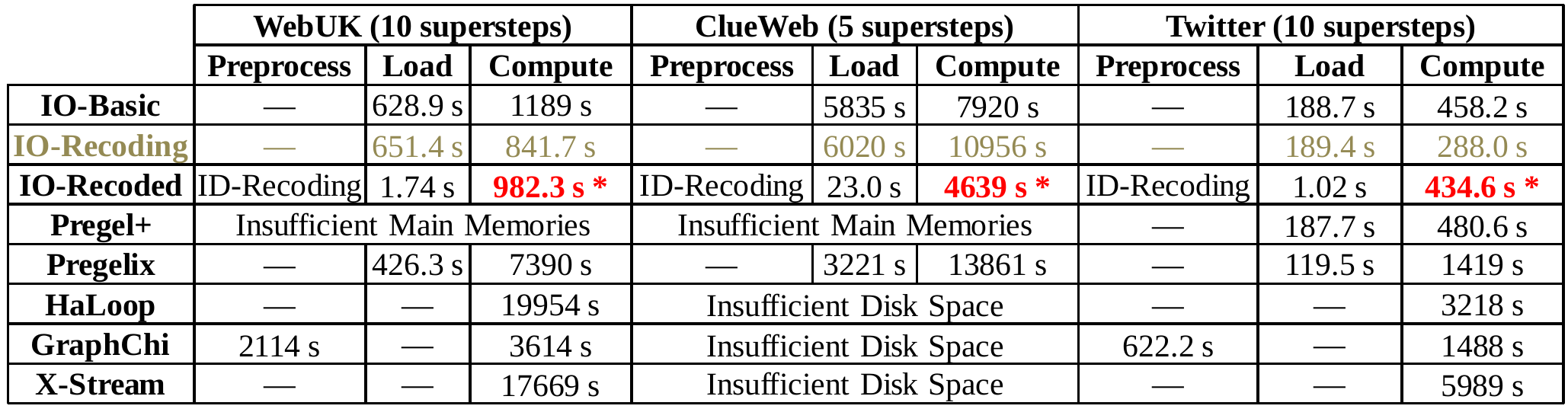}
\caption{Performance of PageRank Computation on $\mathbb{W}^{high}$ (time marked with $\star$: smallest among all systems)}\label{pagerank_high}
\includegraphics[width=1.7\columnwidth]{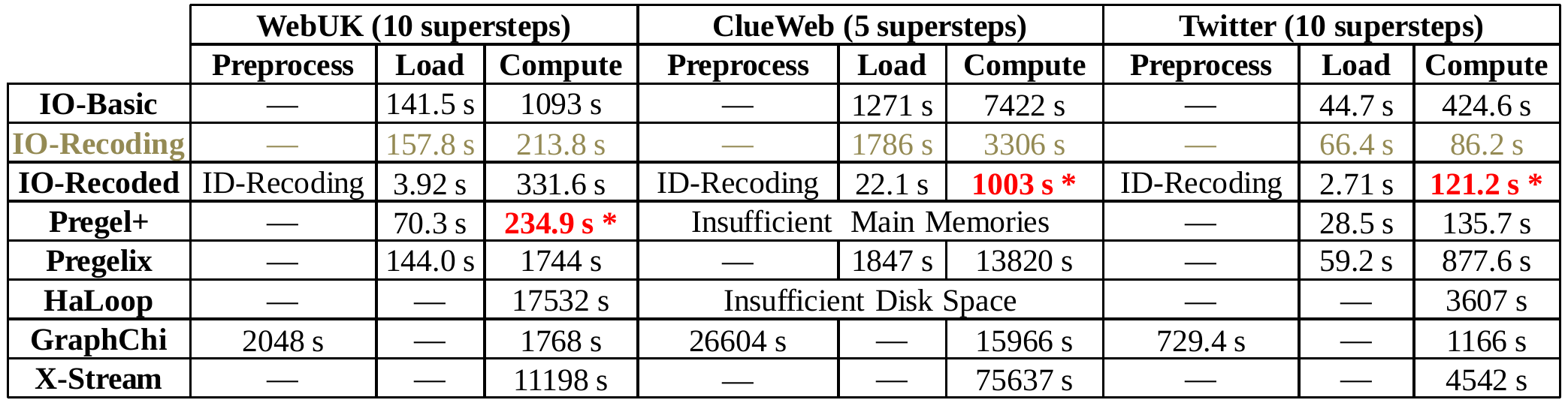}
\end{table*}

Table~\ref{data} lists the five real graph datasets that we used: two directed web graphs {\em WebUK}\footnote{\scriptsize http://law.di.unimi.it/webdata/uk-union-2006-06-2007-05} and {\em ClueWeb}\footnote{\scriptsize http://law.di.unimi.it/webdata/clueweb12}; two social networks {\em Twitter}\footnote{\scriptsize http://konect.uni-koblenz.de/networks/twitter\_mpi} and {\em Friendster}\footnote{\scriptsize http://snap.stanford.edu/data/com-Friendster.html}; and an RDF graph {\em BTC}\footnote{\scriptsize http://km.aifb.kit.edu/projects/btc-2009/}.

Notably, {\em ClueWeb} has 42 billion edges and its input file size exceeds 400GB, and thus single-PC systems can only process it on the machine of $\mathbb{W}^{high}$ that has access to the 2TB disk.

Three Pregel algorithms were used in our evaluation: PageRank and single-source shortest path (SSSP) computation~\cite{pregel} and the Hash-Min algorithm of~\cite{ppa} for computing connected components.

\vspace{2mm}

\noindent{\bf Performance of PageRank.} The experiments were ran on the three directed graphs shown in Table~\ref{data}. We only ran 10 iterations on {\em WebUK} and {\em Twitter} and 5 supersteps on {\em ClueWeb}, since each iteration takes roughly the same time, and while it is efficient for GraphD and Pregel+, it is time-consuming for all the other out-of-core systems that we compared with.

Table~\ref{pagerank_pc} (resp.\ Table~\ref{pagerank_high}) reports the running time of various systems on $\mathbb{W}^{PC}$ (resp.\ $\mathbb{W}^{high}$), where row [IO-Basic] (resp.\ row [IO-Recoded]) reports the performance of the normal mode (resp.\ recoded mode) of GraphD. Row [IO-Recoding] reports the preprocessing time of ID recoding, and we use grey font to differentiate it from other rows that refer to PageRank computation. The other rows report the performance of the systems we compared with, whose header names are self-explanatory.

Column [Load] refers to the time of graph loading. For IO-Basic, Pregel+ and Pregelix, the time is for loading from HDFS; while for IO-Recoded, the time is for loading from local disks (each machine simply loads the recoded state array $A$). In all our tables, an entry ``--'' means ``not applicable''. HaLoop has no loading time since it scans the graph on HDFS in every iteration, and neither do single-PC systems which scan the graph on the local disk in every iteration. Moreover, GraphChi needs to preprocess a graph first by partitioning it into shards, whose time is reported in Column [Preprocess]. Finally, Column [Compute] reports the total time of iterative computation.

From Tables~\ref{pagerank_pc} and~\ref{pagerank_high}, we can see that computation on $\mathbb{W}^{high}$ is much faster than on $\mathbb{W}^{PC}$ due to the more powerful machines and the faster switch. Also, the time of IO-Recoding is consistently less than twice of the data loading time, and is thus an efficient preprocessing. IO-Recoded only slightly improves the performance of IO-Basic on $\mathbb{W}^{PC}$, since $\mathbb{W}^{PC}$ has a lower network bandwidth that forms the bottleneck, and the cost of merge-sort in IO-Basic is mostly hidden inside the cost of message transmission. In contrast, significant improvement is observed on $\mathbb{W}^{high}$ (e.g., over 7 times on {\em ClueWeb}), since IO-Recoded eliminates merge-sort whose cost cannot be fully hidden when the network bandwidth is higher.

As Table~\ref{pagerank_pc} shows, Pregel+ can only process {\em Twitter} on $\mathbb{W}^{PC}$ due to its limited RAM space, and it is even slightly slower than IO-Basic and IO-Recoded. This is because, network bandwidth is the bottleneck in $\mathbb{W}^{PC}$ rather than disk IO, and GraphD's parallel framework fully hides the computation cost inside the communication cost; while in Pregel+'s implementation, message transmission starts after computation finishes (i.e., all messages are generated).

In contrast, as Table~\ref{pagerank_high} shows, Pregel+ is faster than IO-Basic on $\mathbb{W}^{high}$ for both {\em WebUK} and {\em Twitter} since the cost of merge-sort in IO-Basic is not hidden. However, ID-Recoded still beats Pregel+ on {\em Twitter} since the high parallelism of GraphD's execution framework hides the cost of streaming $S^E$ and OMSs.

Among the other systems, Pregelix is much slower than IO-Basic since it performs costly relational operations, and X-Stream is generally much slower than GraphChi as was also observed by~\cite{venus}. However, preprocessing in GraphChi is expensive: sharding {\em ClueWeb} takes 26604 seconds on $\mathbb{W}^{high}$, for which IO-Recoded can already finish over 100 supersteps. Finally, HaLoop is sometimes even slower than X-Stream even though HaLoop uses all machines.

Among the 6 data-cluster combinations reported by Table~\ref{pagerank_pc} and Table~\ref{pagerank_high}, IO-Recoded beats Pregel+ (and is the fastest system) in 5 of them, which is quite amazing given that GraphD is an out-of-core system while Pregel+ is an in-memory system.

\begin{table}[t]
\vspace{-4mm}
\centering
\caption{Message Generation v.s.\ Message Transmission}\label{pagerank_comp}
\includegraphics[width=1.035\columnwidth]{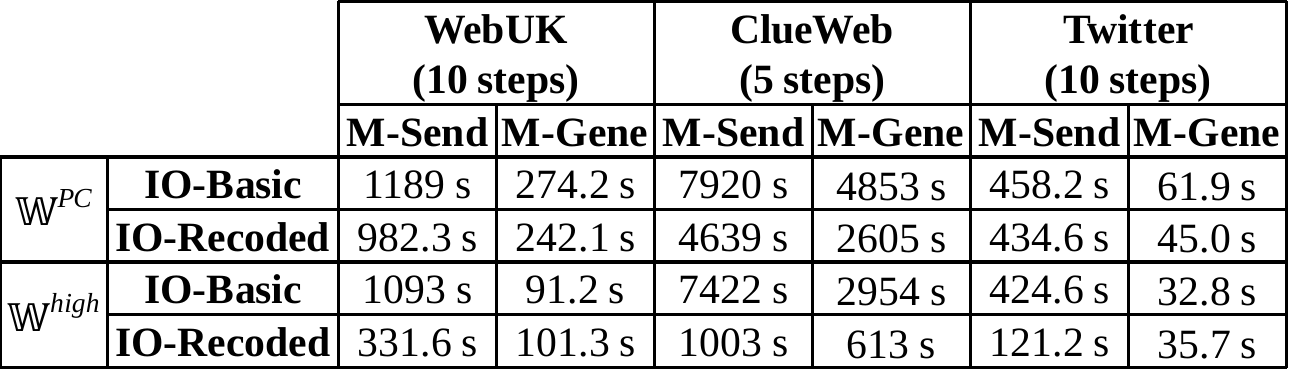}
\vspace{-4mm}
\end{table}

\begin{table*}[t]
\begin{minipage}[t]{0.45\linewidth}
\centering
\caption{Performance of Hash-Min on $\mathbb{W}^{PC}$}\label{hashmin_pc}
\includegraphics[width=1.111\columnwidth]{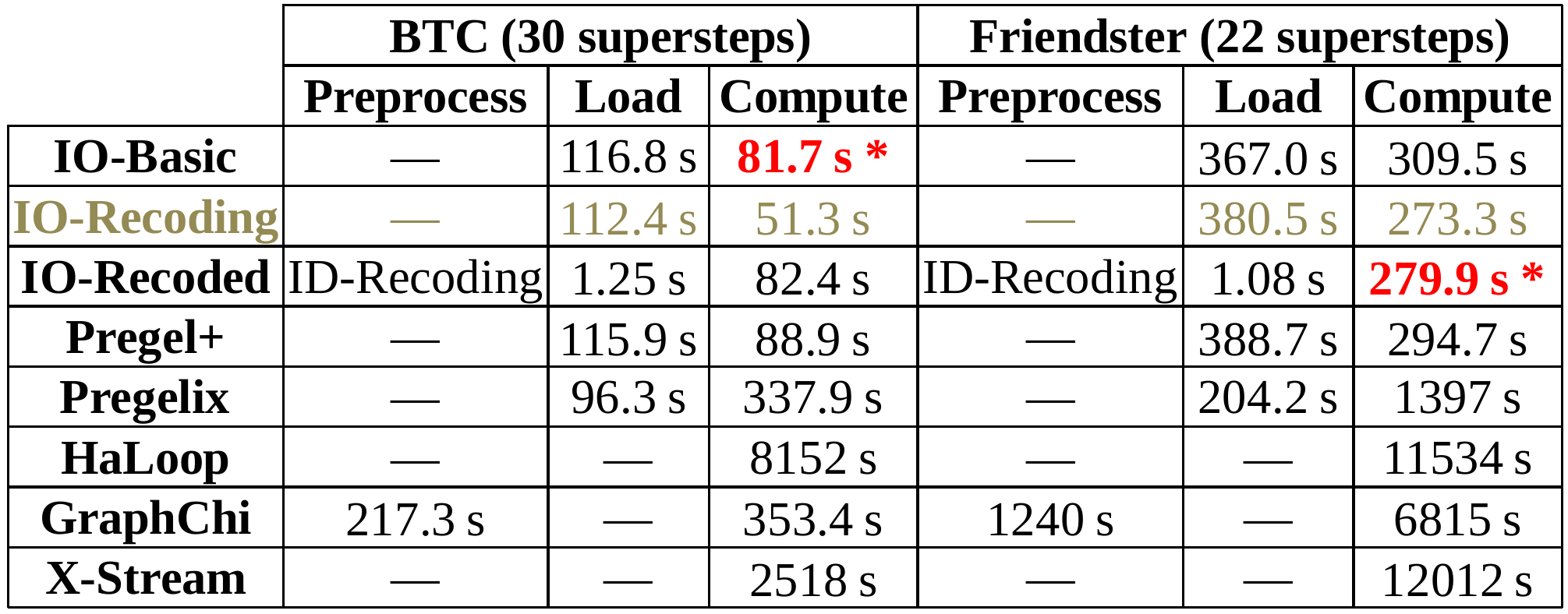}
\end{minipage}
\hfill
\begin{minipage}[t]{0.5\linewidth}
\centering
\caption{Performance of Hash-Min on $\mathbb{W}^{high}$}\label{hashmin_high}
\includegraphics[width=\columnwidth]{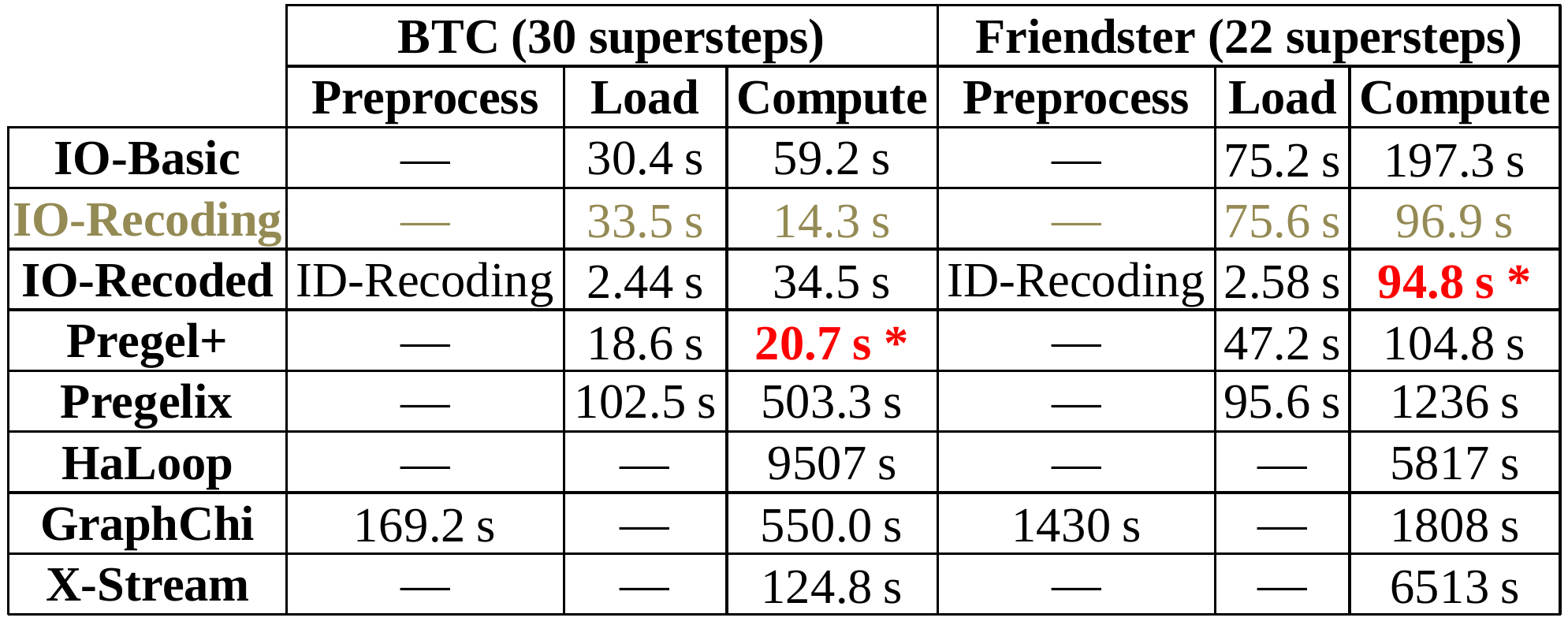}
\end{minipage}
\end{table*}

\vspace{2mm}

\noindent{\bf Message Generation and Transmission Costs.} Table~\ref{pagerank_comp} shows the time taken by both IO-Basic and IO-Recoded to transmit messages (Column [M-Send]), and the fraction of time that $U_c$ spent on generating messages (Column [M-Gene]), for the previous experiments on PageRank computation. Since the behavior of $U_c$ on different machines may vary, we only report the time for $U_c$ on the first machine, and the time is summed over the vertex-centric computation time of all the 10 (or 5) supersteps.

Since network bandwidth is the bottleneck, we can see from Table~\ref{pagerank_comp} that in all the 6 data-cluster combinations, message transmission happens during the whole period of each superstep, but $U_c$ only computes in the early stage (often less than half) of each superstep.

\vspace{2mm}

\noindent{\bf Performance of Hash-Min.} The experiments were ran on the two undirected graphs shown in Table~\ref{data}, and the results are reported in Tables~\ref{hashmin_pc} and~\ref{hashmin_high}. Similar to the experiments on PageRank computation, we can observe that IO-Basic, IO-Recoded and Pregel+ have similar performance on $\mathbb{W}^{PC}$ whose network bandwidth is low, and IO-Recoded even beats Pregel+ over {\em Friendster} on $\mathbb{W}^{high}$.

The computation workload of Hash-Min is typically as follows: most vertices perform computation in the first few supersteps, but as computation goes on, less and less vertices perform computation in a superstep, making the computation workload very sparse. Sparse workload is not a problem for Pregel+ since all adjacency lists are in memories; meanwhile, GraphD is also able to avoid accessing many useless adjacency lists with the help of its streaming function {\em skip}({\em num\_items}) which we introduced in Section~\ref{ssec:edgestream}. However, the other out-of-core systems do not have effective support for sparse workload, and thus as Tables~\ref{hashmin_pc} and~\ref{hashmin_high} show, their computation times are much longer than those of GraphD and Pregel+.

\begin{table*}[t]
\centering
\caption{Performance of SSSP Computation on $\mathbb{W}^{PC}$}\label{sssp_pc}
\includegraphics[width=2.1\columnwidth]{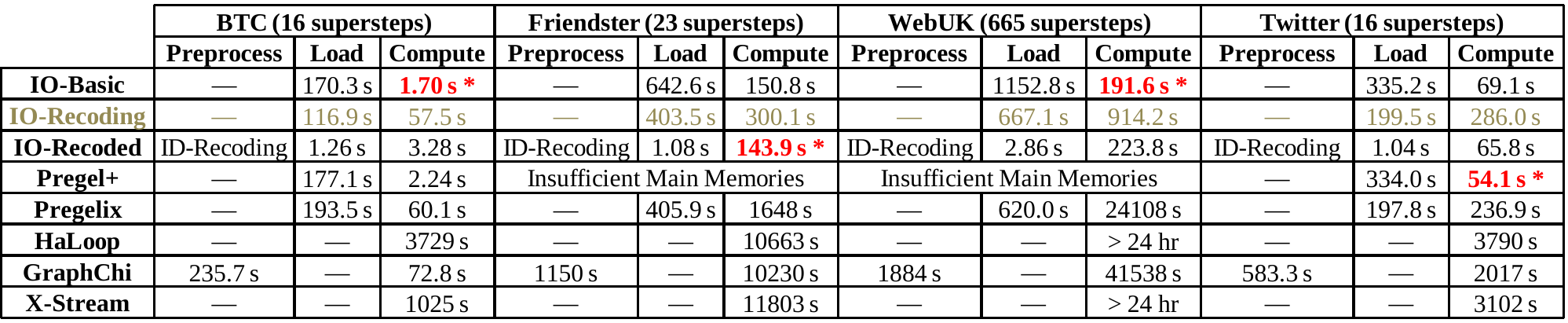}
\caption{Performance of SSSP Computation on $\mathbb{W}^{high}$}\label{sssp_high}
\includegraphics[width=2.1\columnwidth]{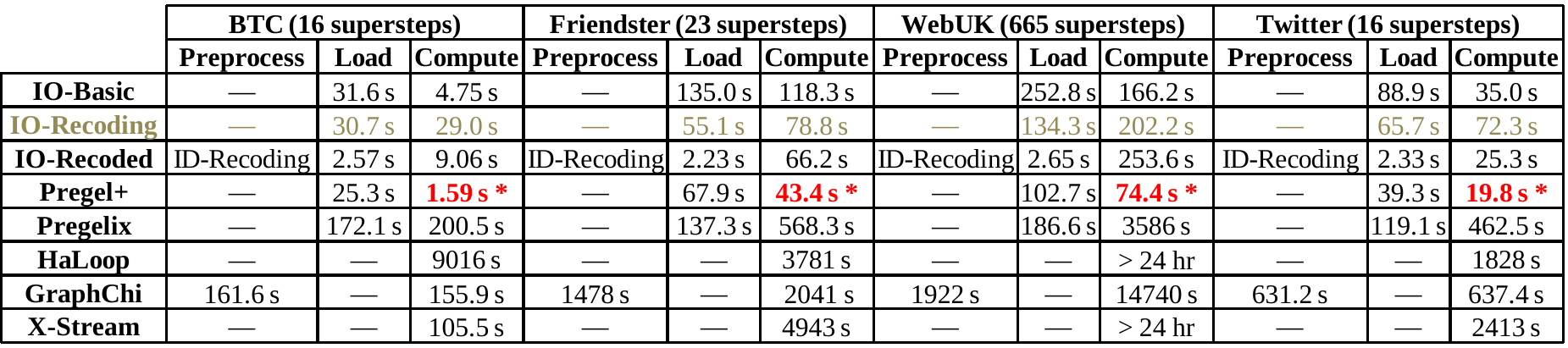}
\end{table*}

\vspace{2mm}

\noindent{\bf Performance of SSSP.} The experiments were ran on the graphs in Table~\ref{data} except for {\em ClueWeb}, for which we could not find a source vertex that can reach to a relatively large amount of other vertices after a long period of trials. All edges were given weight 1, and thus the computation is essentially breadth-first search (BFS).

Unlike PageRank computation and the Hash-Min algorithm discussed before, the computation workload of every superstep is sparse for BFS (or more generally, SSSP). To see this, consider BFS, where a vertex will only send messages to its neighbors when it is reached from the source vertex for the first time. Since every vertex sends messages along adjacent edges for only once during the whole period of computation, the total workload is merely $O(|E|)$, which amounts to the workload of just one superstep in PageRank computation. We remark that BFS (or more generally, SSSP) represents the class of Pregel algorithms that are the most challenging to out-of-core systems which scan disk-resident graphs.

The experimental results are reported in Tables~\ref{sssp_pc} and~\ref{sssp_high}, where we can see that Pregel+ beats all the out-of-core systems in 6 out of the 8 data-cluster combinations. This is, however, not surprising since Pregel+ keeps all adjacency lists in memories. GraphD is not much slower than Pregel+, and even won in 2 data-cluster combinations, thanks to the use of streaming function {\em skip}({\em num\_items}).

Surprisingly, on {\em BTC} and {\em WebUK}, IO-Basic even outperforms IO-Recoded. This is because, if there are too few messages to send in each superstep, the overhead of manipulating the additional arrays (i.e. $A_r$ and $A_s$) in recoded mode backfires. Note that all computations on {\em BTC} finished in seconds for both mode of GraphD, whose workload is really low. While computations on {\em WebUK} took longer time, this is mainly because of the large number of supersteps (i.e., 665). After all, IO-Recode needs to create/update/free those large additional arrays for 665 supersteps.

Also surprisingly, on {\em WebUK},  Pregelix is over two orders of magnitude slower than GraphD on $\mathbb{W}^{PC}$. We found that Pregelix incurs a fixed cost of at least 35 seconds for each superstep, while a superstep of IO-Basic can be as low as 0.02--0.03 seconds. In contrast, Pregelix is much faster on $\mathbb{W}^{high}$ due to faster network, and the fixed cost for a superstep is reduced to 3--4 seconds.

Tables~\ref{sssp_pc} and~\ref{sssp_high} also show that X-Stream is impractical for jobs that run many iterations of sparse-workload vertex computation, since it needs to stream all edges in each iteration. For example, X-Stream could not finish on {\em WebUK} in both $\mathbb{W}^{PC}$ and $\mathbb{W}^{high}$ after a whole day. In fact, the authors of X-Stream themselves admitted this problem at the end of Section~5.3 in~\cite{xstream}.

Finally, graph loading in IO-Recoding is faster than IO-Basic in Tables~\ref{sssp_pc} and~\ref{sssp_high}. This is because during IO-Recoding, $S^{E}$ does not include edge weights. We only attach edge weights when we append recoded adjacency list items to $S^E_{rec}$.

\section{Conclusions}\label{sec:conclude}
We presented an efficient Pregel-like system, called GraphD, for processing very large graphs in a small cluster with ordinary computing resources that are available to most users. To process a graph $G = (V,E)$ with $n$ machines using GraphD, we proved that each machine only requires $O(|V|/n)$ memory space.

While sparse computation workload is not well supported by all previous out-of-core systems, GraphD adopts a new streaming function {\em skip}({\em num\_items}) to handle sparse computation workload efficiently, while attaining sequential I/O bandwidth when the computation workload becomes dense.

For the common cluster setting where machines are connected with Gigabit Ethernet, GraphD fully overlaps computation with communication by buffering outgoing messages to local disks, whose cost is, in turn, hidden inside the cost of message transmission.

When message combining is applicable, GraphD further uses an effective ID-recoding technique to eliminate the need of any expensive external-memory operations such as merge-sort, achieving almost the minimum possible I/O cost that can be expected from any out-of-core Pregel-like system which streams edges and messages on secondary storage.

Open-source implementation of GraphD is provided, and extensive experiments were conducted showing that GraphD's performance is order of magnitude faster than existing out-of-core systems, and is competitive even when compared with an in-memory Pregel-like system running with sufficient memory.

{\small
\bibliographystyle{abbrv}
\bibliography{ref_graphd}
}

\end{document}